\documentclass[letterpaper,10pt,conference]{IEEEtran}
\IEEEoverridecommandlockouts
\usepackage{amsmath,amssymb,mathtools,bm,amsthm}
\usepackage{graphicx}
\usepackage{booktabs,tabularx,array,multirow}
\usepackage{cite}
\usepackage{microtype}
\usepackage{stfloats}
\usepackage{placeins}
\usepackage{balance}
\usepackage{tikz}
\usetikzlibrary{arrows.meta,positioning,calc}
\usepackage{xcolor}

\newtheorem{proposition}{Proposition}
\theoremstyle{remark}
\newtheorem{remark}{Remark}
\newcommand{\R}{\mathbb{R}}
\newcommand{\sat}{\operatorname{sat}}
\newcommand{\proj}{\operatorname{proj}}
\newcommand{\norm}[1]{\left\lVert#1\right\rVert}
\newcolumntype{Y}{>{\raggedright\arraybackslash}X}
\title{Adaptive control of pool-type research nuclear reactors: theory limits, pragmatic actuation constraints, and realized performance and adaptation}

\author{Dimitrios Pylorof \qquad \qquad Linyu Lin%
\thanks{D. Pylorof and L. Lin are with the U.S. Department of Energy Idaho National Laboratory, Idaho Falls, ID 83415, USA. This work was supported by the U.S. Department of Energy under Contract No. DE-AC07-05ID14517. The views expressed do not necessarily reflect those of the U.S. Government.}}

\begin{document}
\maketitle

\begin{abstract}
Adaptive control is attractive for nuclear reactors expected to operate predictably with extensive automation or bounded autonomy under dynamic electric and nonelectric demands, even as reactor and actuator characteristics depart from their nominal values. Standard adaptive control results, however, assume fixed matched dynamics and an input applied as computed. In this work we evaluate a recent model reference adaptive control with linear quadratic regulation (MRAC--LQR) construction on the nonlinear point-kinetics model of Purdue University Reactor Number One (PUR-1), with a moving physical reference, nonlinear rod worth, rate and travel limits, quantization, deadband, and optional power-only state reconstruction. An exact sampled decomposition separates matched adaptation from observer, command processing, actuator realization, independent reactivity, and unmatched-dynamics residuals. Numerical results emphasize compatible reference trajectories, plant parameter variation, response to unexpected rod motion, and actuator duty. Across sampled plants, one unchanged nominal design yields closely aligned responses through adaptively calculated, plant-dependent actuation histories, illustrating both the premise and the theory limits of adaptation at a pool-reactor scale that can support follow-on experimental verification and progression toward higher-consequence nuclear applications.
\end{abstract}


\section{Introduction}
Since the advent of nuclear power generation in 1951 in the desert of eastern Idaho, reactor feedback control has matured significantly, but proposed missions and requirements are changing rapidly. New reactor concepts are being considered for remote sites, islanded grids, industrial process heat, non-baseload electricity generation, and applications that seek predictable operation with extensive automation or bounded autonomy and limited or no onsite staffing; naval and spacecraft propulsion represent notable frontiers. Their power demands may be unusually dynamic relative to the predominantly baseload operation that has characterized most nuclear deployments. In parallel, the reactor and actuators seen by the controller can change with operating conditions and service life through fuel burnup, poison inventory, temperature feedback, and rod worth~\cite{IAEA2018,Bhowmik2024,Ramuhalli2019,Fleming2020,Ramuhalli2025}. While protection and safety systems remain independent of feedback control, a critical question remains: \textit{can the power control layer follow compatible trajectories, reject disturbances, and remain predictable as both the mission and the true plant depart from the nominal design}?

The premise of model reference adaptive control (MRAC) is familiar: select a stable reference behavior, measure how the actual system departs from it, and update controller parameters so that the plant behaves more like the reference. We use the recent Fisher--Annaswamy MRAC--LQR construction~\cite{Fisher2026}, which combines a stabilizing nominal linear quadratic regulator (LQR), a direct adaptive update, and an epoch-gated learn-and-redesign mechanism, while retaining the model-following premise of direct MRAC~\cite{Narendra1989,Astrom1995}. A gap nevertheless emerges for any physical, hardware-actuated system: what happens when the reference moves, uncertainty is not fully matched, and the computed input must pass through a nontrivial actuator? We develop this work as both a complement to and a reality check on adaptive control as an implementable reactor control approach.

Pool-type research reactors provide an intermediate platform for examining adaptive control in nuclear systems. They retain neutron kinetics, control rods, instrumentation, operating procedures, and independent protection without the scale and regulatory complexity of a commercial plant. PUR-1 also provides a digital control and cyber-physical setting intended for staged evaluation~\cite{Theos2026}. This role resembles NASA's use of the AirSTAR subscale aircraft to confront adaptive flight control designs with sensing and actuation before higher-consequence deployment~\cite{Gregory2010}.

Adaptive control in nuclear systems has been studied through state and output feedback MRAC for reactor-power tracking under parametric uncertainty and external disturbances~\cite{reddy2020state,reddy2023output}, and through model-reference robust adaptive control of control-element drive mechanisms~\cite{park2020cedm}. Related LQR-referenced adaptive control has also been applied to a glycol heat exchanger in an integrated energy system, representing balance-of-plant thermal equipment rather than reactor core dynamics~\cite{seurin2025thermal}. The present work links these lines by evaluating core power adaptation through an explicit nonlinear, sampled, and constrained control rod interface.

Our contributions are as follows:
\begin{enumerate}
    \item We implement MRAC--LQR for PUR-1 load following and evaluate nonlinear and constrained actuation, quantization, deadband, observer-based state reconstruction, and disturbances, including effects outside the assumptions of the underlying theory.
    \item We formulate an exact sampled decomposition that separates matched adaptation from moving-reference, observer, actuator, independent reactivity, and unmatched dynamics residuals.
    \item Numerical evaluations cover direct and shaped maneuvers, 30 plant cases with paired adaptive/fixed-nominal runs, unexpected motion of a second rod, power-only reconstruction, and the deadband trade between power precision and actuator duty.
\end{enumerate}
Together, the formulation and 30-case comparison illustrate a practical meaning of robustness for adaptive control. The nominal physical reference, target dynamics, initial LQR gain, adaptive law structure, actuator calibration, and related logic are fixed once. Each plant is then simulated independently from the same controller initialization, with online adaptation driven only by that plant's observed transitions and delivered input. Despite different plant dynamics, the power trajectories remain closely aligned while the rod histories differ. Unlike a conventional plant family study with a fixed controller, the coefficients of each deployed plant are not used for plant-specific synthesis; the common adaptive architecture responds online to the plant encountered. This distinction is relevant as reactor parameters evolve over the life cycle.

\section{Reactor, Reference, and Controller}
\subsection{Physical plant and reactivity-driven dynamics}
Let $\theta\in\Theta$ collect physical reactor and actuator coefficients used in the numerical evaluations. The reactor state is
\[
 x=[p,c_1,c_2,c_3,c_4,c_5,c_6,T,I,X]^{\mathsf T}\in\R^{10},
\]
where $p$ is normalized neutron power, $c_i$ are six delayed neutron precursor states, $T$ is a lumped fuel/moderator temperature, and $I$ and $X$ are iodine and xenon states. Define $e_p:=[1\;0\;\cdots\;0]^{\mathsf T}\in\R^{10}$; the measured power is $y=e_p^{\mathsf T}x+\nu=p+\nu$, where $\nu$ is measurement error.

The manipulated physical input $u(t)$ is the AMDR reactivity contribution and $w(t)$ is the independent regulating rod (RR) contribution, both in parts per hundred thousand (pcm), with $1\ {\rm pcm}=10^{-5}$ dimensionless reactivity. Rod positions and the nonlinear maps that generate $u$ and $w$ are defined in Sec.~II-B. Downstream of those maps, the physical reactor is
\begin{equation}
 P_\theta:\qquad \dot x=f_\theta(x,u,w),\qquad
 y=e_p^{\mathsf T}x+\nu.
 \label{eq:plant}
\end{equation}
Two safety shim rods remain fixed during each run; their combined worth is $W_{{\rm fix},\theta}$. With baseline excess reactivity offset $\rho_b$ in pcm, temperature and xenon feedback coefficients $\alpha_T$ and $\alpha_X$, and bulk pool water heat sink temperature $T_\infty$, the dimensionless reactivity entering point kinetics is
\begin{equation}
 \begin{aligned}
 \rho_\theta(x,u,w)=10^{-5}\big[&W_{{\rm fix},\theta}+u+w-\rho_b\\[-2pt]
 &+\alpha_T(T-T_\infty)+\alpha_X X\big].
 \end{aligned}
 \label{eq:reactivity}
\end{equation}
Let $\beta_i$ be the delayed neutron fractions, $\beta=\sum_{i=1}^6\beta_i$, $\Lambda$ the prompt generation time, and $\lambda_i$ the precursor decay constants. Let $\tau_T$ be the thermal time constant; $\lambda_I,\lambda_X$ the iodine and xenon decay constants; $\sigma_X$ the xenon absorption coefficient; and $k_I=y_I\kappa_f$, $k_X=y_X\kappa_f$, where $y_I,y_X$ are fission yields and $\kappa_f$ is the normalized fission rate factor. Then $f_\theta$ is
\begin{subequations}\label{eq:reactor_dynamics}
\begin{align}
 \dot p&=\frac{\rho_\theta(x,u,w)-\beta}{\Lambda}p
       +\sum\nolimits_{i=1}^{6}\lambda_i c_i, \label{eq:pke_p}\\
 \dot c_i&=\frac{\beta_i}{\Lambda}p-\lambda_i c_i,
       \quad i=1,\dots,6, \label{eq:pke_ci}\\
 \dot T&=(35p+T_\infty-T)/\tau_T, \label{eq:dyn_T}\\
 \dot I&=k_Ip-\lambda_I I, \label{eq:dyn_I}\\
 \dot X&=k_Xp+\lambda_I I-(\lambda_X+\sigma_Xp)X. \label{eq:dyn_X}
\end{align}
\end{subequations}
Dependence of the individual coefficients on $\theta$ is suppressed for readability.

\subsection{Rod worth, command processing, and sampled actuation}
The PUR-1 model contains four rod worth channels. The \textit{auxiliary moderate displacement rod} (AMDR) is the feedback actuator; the RR is normally fixed and is moved only as an independent anomaly input; and safety shim rods SS1 and SS2 remain fixed. For rod $j\in\{A,{\rm RR},{\rm SS1},{\rm SS2}\}$, let $q_j$ denote position in centimeters and
\begin{equation}
 W_{j,\theta}(q_j)=a_{j0}(\theta)+a_{j1}(\theta)q_j
 +a_{j2}(\theta)q_j^2+a_{j3}(\theta)q_j^3
 \label{eq:worthmap}
\end{equation}
its cubic worth map in pcm, with $a_{jr}(\theta)$ the polynomial coefficients. A superscript $0$ denotes a fixed or initial position. Thus $q_{\rm SS1}^{0}=54$ cm and $q_{\rm SS2}^{0}=56$ cm, and
\begin{equation}
 W_{{\rm fix},\theta}=W_{{\rm SS1},\theta}(q_{\rm SS1}^{0})
 +W_{{\rm SS2},\theta}(q_{\rm SS2}^{0}).
 \label{eq:fixedworth}
\end{equation}
Let $\bar\theta$ denote the nominal design coefficients and $\overline W_A(q):=W_{A,\bar\theta}(q)$ the calibration retained by the controller. It is strictly increasing on $0\le q\le64$ cm, so its inverse is single valued on the reachable worth interval. The sampled physical AMDR variations preserve this input sign.

At time $t_k=k\,\mathrm{s}$, let $x_k:=x(t_k)$ and $q_{j,k}:=q_j(t_k)$. Let $u_{r,k}$ denote the sampled reference AMDR worth constructed in Sec.~II-C and $v_k^\star$ the dimensionless residual command produced in Sec.~II-E. With $\sat_{[a,b]}(r):=\min\{b,\max\{a,r\}\}$, the actuator-facing command is
\begin{equation}
    v_k^{\rm raw}=\underset{[-v_{\max},v_{\max}]}{\sat}(v_k^\star), \;\; v_k^c=\alpha_fv_k^{\rm raw}+(1-\alpha_f)v_{k-1}^c, \label{eq:filter}
\end{equation}
where $v_{\max}=2$, $\alpha_f=0.45$, and one unit of $v$ is $s_v=10$ pcm. The filtered total worth request and corresponding rod position target are $u_k^{\rm targ}=u_{r,k}+s_vv_k^c$ and $q_{A,k}^{\rm targ}=\mathcal I_A(u_k^{\rm targ})$, where $\mathcal I_A$ inverts $\overline W_A$ when the requested worth is reachable and otherwise returns the nearest travel boundary. The map $\mathcal A_A$ then applies target deadband, rate and travel limits, and 0.1 cm quantization. Let $q_{A,k}^{+}$ and $q_{R,k}^{+}$ be the post-command positions held on $[t_k,t_{k+1})$, and let $F_\theta$ be the one-second flow map of~\eqref{eq:plant}. The sampled physical interconnection is
\begin{align}
 q_{A,k}^{+}&=\mathcal A_A(q_{A,k},q_{A,k}^{\rm targ}),
 &u_k^{\rm ind}&=\overline W_A(q_{A,k}^{+}),\nonumber\\
 u_k&=W_{A,\theta}(q_{A,k}^{+}),
 &w_k&=W_{{\rm RR},\theta}(q_{R,k}^{+}),\nonumber\\
 x_{k+1}&=F_\theta(x_k,u_k,w_k),
 &q_{A,k+1}&:=q_{A,k}^{+}.
 \label{eq:actchain}
\end{align}
The RR memory is updated analogously; normally $q_{R,k}^{+}=q_{R,k}$. The physical plant receives $u_k$, whereas adaptation uses the delivered, nominally indicated residual
\begin{equation}
 v_k^a=(u_k^{\rm ind}-u_{r,k})/s_v.
 \label{eq:deliveredresidual}
\end{equation}
\begin{remark}[Actuator cascade]
The effective system seen by adaptation includes limiting, filter memory, nominal worth inversion, rod realization, the physical worth map, and reactor dynamics. Calibration mismatch is compensable only to the extent that it preserves sign and authority and is representable through the AMDR input direction; nonlinear mismatch may instead create time-varying effective coefficients. Clipping, filtering, rate and travel limits, deadband, and quantization alter command realization and remain outside the fixed matched-input result, so their effects are evaluated numerically.
\end{remark}

\subsection{Compatible physical reference trajectory}
Upstream of the adaptive control mechanics, it is important to differentiate between the external power command $p_c$ (potentially piecewise continuous), desired power $p_r$, physical reference $(x_r,u_r)$, and the tracking error dynamics. In routine load following, an upstream scheduler or supervisor should provide a trajectory compatible with operating limits, ramp capability, and mission objectives; MRAC--LQR regulates to that trajectory rather than constructing it. Higher level reference governor and learning-based trajectory generation have been considered for nuclear systems~\cite{dave2023design,pylorof2022reinforcement,pylorof2026unified}. The step commands used in this work are diagnostics. We either set $p_r=p_c$ or pass $p_c$ through a fixed stable filter $G_r$ selected offline as a guard against discontinuous requests. Unlike a formal reference governor, $G_r$ does not recompute the command from state and constraint predictions~\cite{Gilbert1995,dave2023design}.

Let $f_{\bar\theta}$ be the nominal reactor vector field, $q_R^{0}$ the fixed reference RR position, and $w_r:=W_{{\rm RR},\bar\theta}(q_R^{0})$. The physical reference reactor satisfies
\begin{equation}
 \dot x_r=f_{\bar\theta}(x_r,u_r,w_r),
 \qquad e_p^{\mathsf T}x_r=p_r.
 \label{eq:referenceplant}
\end{equation}
The non-power components of $x_r$ follow~\eqref{eq:pke_ci}--\eqref{eq:dyn_X} with nominal coefficients. Inverse kinetics gives
\begin{equation}
 \rho_r=\beta+\frac{\Lambda}{p_r}
 \left(\dot p_r-\sum_{i=1}^{6}\lambda_i c_{i,r}\right),
 \label{eq:rhoref}
\end{equation}
and the compatible reference AMDR worth is
\begin{equation}
 u_r=10^5\rho_r-W_{{\rm fix},\bar\theta}-w_r+\rho_b
 -\alpha_T(T_r-T_\infty)-\alpha_X X_r,
 \label{eq:uref}
\end{equation}
with nominal coefficients on the right-hand side. The corresponding reference position is $q_{A,r}=\overline W_A^{-1}(u_r)$. At $t_k$, define $x_{r,k}:=x_r(t_k)$ and $u_{r,k}:=u_r(t_k)$. Hence $p_c\mapsto p_r\mapsto(x_r,u_r,q_{A,r})$ defines the moving physical reference.

\subsection{Tracking coordinates and nominal error model}
Let $x_{\rm eq}^{\rm nom}=[p_{\rm eq},c_{1,{\rm eq}},\ldots,c_{6,{\rm eq}},T_{\rm eq},I_{\rm eq},X_{\rm eq}]^{\mathsf T}$ be the nominal equilibrium at the initial power level, and define
\begin{equation}
 S_x=\operatorname{diag}\!\left(
 p_{\rm eq}^{-1},c_{1,{\rm eq}}^{-1},\ldots,c_{6,{\rm eq}}^{-1},
 0.1,I_{\rm eq}^{-1},X_{\rm eq}^{-1}\right),
 \label{eq:scaling}
\end{equation}
where the temperature entry is in K$^{-1}$. The physical tracking coordinate and its integral augmentation are
\begin{equation}
 z_k:=S_x(x_k-x_{r,k})\in\R^{10},\qquad
 \bar z_k:=\begin{bmatrix}z_k\\ \eta_k\end{bmatrix}\in\R^{11},
 \label{eq:trackingerror}
\end{equation}
where $\eta_k$ is the bounded integral of scaled power error and, when antiwindup is inactive, $\eta_{k+1}=\eta_k+e_p^{\mathsf T}z_k$.

Linearizing the nominal plant/reference pair in the sampled residual coordinate $v_k=(u_k^{\rm ind}-u_{r,k})/s_v$ and applying exact one-second zero-order-hold discretization gives
\begin{equation}
 z_{k+1}=A_pz_k+B_pv_k,
 \quad A_p\in\R^{10\times10},\quad B_p\in\R^{10}.
 \label{eq:nominalerror}
\end{equation}
The augmented pair is
\begin{equation}
 A_0=\begin{bmatrix}A_p&0\\e_p^{\mathsf T}&1\end{bmatrix},\;
 B_0=\begin{bmatrix}B_p\\0\end{bmatrix},
 \; \bar z_{k+1}=A_0\bar z_k+B_0v_k.
 \label{eq:augpair}
\end{equation}
All controllers use
\begin{equation*}
 Q=\operatorname{diag}\!\left(20,\underbrace{0.7,\ldots,0.7}_{6},2,0.03,0.08,0.10\right)\in\R^{11\times11},
\end{equation*}
$R=[4]\in\R_{>0}$, $K_0=\operatorname{dlqr}(A_0,B_0,Q,R)\in\R^{1\times11}$, $B_m=B_0$, and $A_m=A_0+B_0K_0$, where $\operatorname{dlqr}$ denotes the discrete time LQR gain and $A_m$ is Schur. The entries of $Q$ weight, in order, scaled power, six precursors, temperature, iodine, xenon, and integral error. The Jacobian and $K_0$ are evaluated at 10\% power except for the 40\% multilevel maneuver; $K_0$ is fixed within each comparison.

\subsection{MRAC--LQR on the tracking system}
We apply the Fisher--Annaswamy MRAC--LQR algorithm~\cite{Fisher2026} to the tracking error system. Let $\widehat a_k\in\R^{11}$, $\widehat b_k\in\R$, and $\widehat g_k\in\R$ denote its adaptive feedback, input effectiveness, and matched bias coefficients. They are effective control quantities, not real-time estimates of the physical coefficients $\theta\in\Theta$. Let $\bar z_k^{\rm fb}$ denote the coordinate supplied by the feedback path: $\bar z_k^{\rm fb}=\bar z_k$ under full state feedback and, under power-only feedback, it is the reconstructed coordinate defined in Sec.~II-F. We initialize $(\widehat a_0,\widehat g_0,\widehat b_0)=(K_0^{\mathsf T},0,1)$ and project componentwise to
\begin{equation}
 \begin{aligned}
 K_0^{\mathsf T}-2\bm{1}_{11}&\le\widehat a_k\le K_0^{\mathsf T}+2\bm{1}_{11},\\
 0.30&\le\widehat b_k\le1.80,\qquad |\widehat g_k|\le1.
 \end{aligned}
 \label{eq:parameterbox}
\end{equation}
where $\bm{1}_{11}$ is the vector of ones. At redesign epoch $\ell$, let $\Delta_{m,\ell}\in\R^{11}$ be the epoch correction, initially zero. The unconstrained residual worth command is
\begin{equation}
 v_k^\star=\widehat b_k^{-1}
 \left[(\widehat a_k+\Delta_{m,\ell})^{\mathsf T}\bar z_k^{\rm fb}-\widehat g_k\right].
 \label{eq:control}
\end{equation}
The integral state satisfies $|\eta_k|\le0.50$ and is frozen under active AMDR rate or travel limitation; exploratory excitation is disabled. Define $\widehat\psi_k=[\widehat a_k^{\mathsf T}\ \widehat g_k\ \widehat b_k]^{\mathsf T}\in\R^{13}$, $\phi_k=[-(\bar z_k^{\rm fb})^{\mathsf T}\ 1\ v_k^a]^{\mathsf T}\in\R^{13}$, and $\zeta_{k+1}:=B_m^\dagger(\bar z_{k+1}^{\rm fb}-A_m\bar z_k^{\rm fb})\in\R$, where $\dagger$ denotes the Moore--Penrose pseudoinverse. With Euclidean projection $\proj$ onto~\eqref{eq:parameterbox}, the fast update is
\begin{equation}
 \widehat\psi_{k+1}=\proj\!\left[
 \widehat\psi_k+
 \frac{(\zeta_{k+1}-\widehat\psi_k^{\mathsf T}\phi_k)\phi_k}
 {\max\{1,\norm{\phi_k}^2\}}\right].
 \label{eq:update}
\end{equation}
Equation~\eqref{eq:update} updates the applied coefficients at 1~Hz within prescribed bounds. Recursive least squares (RLS) runs at the same rate, maintaining a separate parameter estimate and information matrix from the observed transitions. With the selected box projection, the RLS estimate enters neither~\eqref{eq:control} nor~\eqref{eq:update}. Its accumulated information gates LQR redesign, which is eligible only after an information increment of $2\times10^{-4}$ and a minimum age of $900(\ell+1)$ samples. A candidate redesign uses $(\widehat a_k,\widehat b_k)$ and must pass a Schur stability check on the resulting linear closed loop. Thus ``slower'' refers to gated redesign, while both learning recursions run at 1~Hz. In the plant family runs, no redesign is accepted: $\Delta_{m,\ell}=0$ and the nominal LQR design is retained, while direct adaptation changes the effective feedback law.

Before each plant run, $\widehat\psi_0$, $\eta_0$, filter state, and RLS/epoch memory are reset. In the fixed-nominal pair,~\eqref{eq:control}, integral action, filtering, rod realization, and the nonlinear plant remain active, while~\eqref{eq:update}, RLS, and redesign are skipped; hence $\widehat\psi_k=\widehat\psi_0$ and $\Delta_{m,\ell}=0$. This isolates online adaptation from nominal feedback and actuator dynamics.

The complete sampled loop therefore includes more state than~\eqref{eq:plant}. With $\operatorname{col}(\cdot)$ denoting column concatenation and $\mathcal R_k$ the RLS/epoch memory, define
\begin{equation}
 \chi_k=\operatorname{col}(x_k,x_{r,k},q_{A,k},q_{R,k},
 \eta_k,v_{k-1}^c,\widehat\psi_k,\mathcal R_k).
 \label{eq:closedloopstate}
\end{equation}
With $p_{c,k}:=p_c(t_k)$ and $\nu_k:=\nu(t_k)$, let $\omega_{R,k}$ be an exogenous RR motion request, zero in normal operation. Then $\chi_{k+1}=\mathcal F_{{\rm cl},\theta}(\chi_k,p_{c,k},\nu_k,\omega_{R,k})$ for a closed-loop map $\mathcal F_{{\rm cl},\theta}$. Power-only operation augments $\chi_k$ with $\widehat x_k$. MRAC--LQR acts on $\bar z_k^{\rm fb}=\bar z^{\rm fb}(\chi_k)$ induced by this interconnection, not directly on the isolated state $x_k$.

\subsection{Power-only state reconstruction}
Reactor instrumentation typically does not directly measure many of the states in $x$. With power-only feedback, measured power supplies the first state and the remaining nine states are propagated from reactor physics. Over each completed one-second interval, let $y(t)$ linearly interpolate the endpoint measurements. The estimates satisfy
\begin{align}
 \dot{\widehat c}_i&=\frac{\beta_i}{\Lambda}y-\lambda_i\widehat c_i,
 &\dot{\widehat T}&=\frac{35y+T_\infty-\widehat T}{\tau_T},\nonumber\\
 \dot{\widehat I}&=k_Iy-\lambda_I\widehat I,
 &\dot{\widehat X}&=k_Xy+\lambda_I\widehat I-(\lambda_X+\sigma_Xy)\widehat X.
 \label{eq:observer}
\end{align}
Thus $\widehat x=[y,\widehat c_1,\ldots,\widehat c_6,\widehat T,\widehat I,\widehat X]^{\mathsf T}$, $\widehat z_k=S_x(\widehat x_k-x_{r,k})$, and $\widehat{\bar z}_k=[\widehat z_k^{\mathsf T}\ \eta_k]^{\mathsf T}$. Consequently, $\bar z_k^{\rm fb}=\widehat{\bar z}_k$ in power-only operation. This is a physics propagation observer, \textit{not} a Luenberger observer.

For estimate-minus-truth errors, e.g., $\widetilde X=\widehat X-X$, subtraction of~\eqref{eq:reactor_dynamics} from~\eqref{eq:observer} gives
\begin{equation}
 \begin{aligned}
 \dot{\widetilde c}_i&=-\lambda_i\widetilde c_i+(\beta_i/\Lambda)\nu, \;\;\;\;
 \dot{\widetilde T}=-\widetilde T/\tau_T+35\nu/\tau_T,\\
 \dot{\widetilde I}&=-\lambda_I\widetilde I+k_I\nu,\\
 \dot{\widetilde X}&=\lambda_I\widetilde I-\lambda_X\widetilde X+k_X\nu
 -\sigma_X(y\widetilde X+X\nu).
 \end{aligned}
 \label{eq:observererror}
\end{equation}
With exact power measurement and consistent initialization, zero estimation error is invariant. Precursor and temperature errors decay much faster than iodine and xenon errors. Three initialization levels are used below, always with the physical reactor at its true 10\% equilibrium and measured power initialized exactly. The \emph{warm} case sets all nine hidden states to their true equilibrium values. The \emph{intermediate} (semi-warm) case keeps the six precursors, temperature, and iodine exact but sets $\widehat X_0=10X_0$, isolating a large slow-state error. The \emph{observer-cold} case sets $\widehat c_{i,0}=\widehat I_0=\widehat X_0=0$ and $\widehat T_0=T_\infty$; ``cold'' refers to estimator memory, not the physical reactor. Because the reconstructed coordinate enters both~\eqref{eq:control} and~\eqref{eq:update}, the full-state MRAC--LQR result does not provide a separation principle for the coupled observer/adaptive loop.
\section{Exact Error Decomposition and Theory Boundary}
\subsection{Moving-reference sampled dynamics}
For the physical reference in Sec.~II-C, load following is regulation of the tracking coordinate $z_k$ in~\eqref{eq:trackingerror} about zero. Combining the sampled rod/plant interconnection~\eqref{eq:actchain}, the nominal reference reactor~\eqref{eq:referenceplant}, and the delivered residual~\eqref{eq:deliveredresidual}, the smooth one-second plant/reference flow admits an exact incremental factorization along the realized transition:
\begin{equation}
 z_{k+1}=A_k z_k+B_kv_k^a+d_k.
 \label{eq:sampled}
\end{equation}
Here $A_k\in\R^{10\times10}$ is the realized state-error transition, $B_k\in\R^{10}$ maps the delivered, nominally indicated AMDR residual into the next error, and $d_k\in\R^{10}$ collects the remaining increment, including the independent reactivity deviation $w_k-w_r$. These are effective transition quantities induced by the nonlinear interconnection and may vary with operating point, reference motion, actuator realization, and the connected plant. Equation~\eqref{eq:sampled} is therefore \textit{not} the fixed linear time-invariant model assumed by the MRAC--LQR result~\cite{Fisher2026}. Partition $K_0=[K_z\ k_\eta]$, with $K_z\in\R^{1\times10}$ and $k_\eta\in\R$, and define the ten-state comparator $A_{m,p}:=A_p+B_pK_z$. Let $P_B:=B_pB_p^\dagger$ be the orthogonal projector onto $\operatorname{im}B_p$, $I_{10}$ the identity, and $\Delta A_k:=A_k-A_{m,p}$. Let $a_k^\star\in\R^{10}$, $b_k^\star,g_k^\star\in\R$, $E_{A,k}\in\R^{10\times10}$, and $E_{B,k},d_k^\perp\in\R^{10}$ be defined by the exact projections
\begin{align}
 a_k^{\star\mathsf T}&=-B_p^\dagger\Delta A_k,
 &E_{A,k}&=(I_{10}-P_B)\Delta A_k,\nonumber\\
 b_k^\star&=B_p^\dagger B_k,
 &E_{B,k}&=(I_{10}-P_B)B_k,\nonumber\\
 g_k^\star&=B_p^\dagger d_k,
 &d_k^\perp&=(I_{10}-P_B)d_k
 \label{eq:projection}
\end{align}
produce
\begin{equation}
 \begin{aligned}
 A_k&=A_{m,p}-B_pa_k^{\star\mathsf T}+E_{A,k},\\
 B_k&=B_pb_k^\star+E_{B,k},\qquad
 d_k=B_pg_k^\star+d_k^\perp.
 \end{aligned}
 \label{eq:decomp}
\end{equation}
The starred quantities are instantaneous control-referenced projections, not physical entries of $\theta$. If $b_k^\star$ has known sign and is bounded away from zero, define
\begin{equation}
 v_k^{\rm id}:=(b_k^\star)^{-1}
 \left(a_k^{\star\mathsf T}z_k-g_k^\star\right).
 \label{eq:ideal}
\end{equation}

\begin{proposition}[Exact matched/unmatched separation]
Definitions~\eqref{eq:projection}--\eqref{eq:ideal} give
\begin{equation}
 z_{k+1}=A_{m,p}z_k+B_pb_k^\star(v_k^a-v_k^{\rm id})+r_k^\perp,
 \label{eq:coreidentity}
\end{equation}
where $r_k^\perp:=E_{A,k}z_k+E_{B,k}v_k^a+d_k^\perp$ and $B_p^\dagger r_k^\perp=0$.
\end{proposition}
\begin{proof}
Substitute~\eqref{eq:decomp} into~\eqref{eq:sampled}, add and subtract $B_pb_k^\star v_k^{\rm id}$, and use~\eqref{eq:ideal}.
\end{proof}

\begin{remark}[Why the projection matters]
The decomposition is algebraically elementary, but the projection is onto $\operatorname{im}B_p$, the only direction through which the scalar AMDR residual can change the nominal physical tracking dynamics. If $r_k^\perp=0$ and $b_k^\star\neq0$, $v_k^{\rm id}$ is the unique delivered residual that realizes $z_{k+1}=A_{m,p}z_k$; if $r_k^\perp\neq0$, no scalar AMDR input cancels that component at the same sample. Thus~\eqref{eq:coreidentity} distinguishes uncertainty structure from input realization and prevents interpreting $a_k^\star,b_k^\star,g_k^\star$ as physical parameter estimates. Saturation or deadband may leave $r_k^\perp$ small yet enlarge $v_k^a-v_k^{\rm id}$. Fisher and Annaswamy's Theorem~2~\cite{Fisher2026} is proved for the fixed discrete-time LTI plant in their Eq.~(1) under Assumptions~1--2: full state measurement, exact matching to a known Schur pair, bounded and invertible input effectiveness, and direct application of the computed input. The directly covered, bias-free counterpart in our notation has constant $a_k^\star,b_k^\star$, $g_k^\star=0$, $r_k^\perp=0$, $\widehat z_k=z_k$, and $v_k^a=v_k^\star$; the matched bias and actuator cascade terms are additional.
\end{remark}

\subsection{Implementation residuals and conditional bound}
Let $v^\star(\bar z,\widehat\psi)$ denote~\eqref{eq:control}. The delivered input mismatch in~\eqref{eq:coreidentity} has the exact telescoping split
\begin{align}
 v_k^a-v_k^{\rm id}
 =&\underbrace{v^\star(\bar z_k,\widehat\psi_k)-v_k^{\rm id}}_{\delta_k^{\rm par}}
 +\underbrace{v^\star(\widehat{\bar z}_k,\widehat\psi_k)-v^\star(\bar z_k,\widehat\psi_k)}_{\delta_k^{\rm obs}}\nonumber\\
 &\!\!\!\!\!+\underbrace{v_k^{\rm raw}-v^\star(\widehat{\bar z}_k,\widehat\psi_k)}_{\delta_k^{\rm mag}}
 +\underbrace{v_k^c-v_k^{\rm raw}}_{\delta_k^{\rm filt}}
 +\underbrace{v_k^a-v_k^c}_{\delta_k^{\rm act}}.
 \label{eq:residualsplit}
\end{align}
These terms represent \textit{adaptive parameter}, \textit{observer}, \textit{magnitude limit}, \textit{filter}, and \textit{command realization} errors. Physical worth map mismatch instead changes $B_k$ and therefore $b_k^\star$ and $E_{B,k}$; it is not hidden inside $\delta_k^{\rm act}$. To make the last term explicit, let $d_A\ge0$ be the target deadband, $D_d(r)=0$ for $|r|<d$ and $D_d(r)=r$ otherwise, $\Pi_{[0,64]}$ the travel projection, and $\mathcal Q_{0.1}$ the 0.1 cm quantizer. With requested increment $\Delta q_{A,k}^{\rm targ}:=q_{A,k}^{\rm targ}-q_{A,k}$ and delivered pre-quantization increment $\Delta q_{A,k}^{\rm del}$, we obtain $\Delta q_{A,k}^{\rm del}=\sat_{[-5,5]}\!\left(D_{d_A}(\Delta q_{A,k}^{\rm targ})\right)$ and $q_{A,k}^{+}=\mathcal Q_{0.1}\!\left(\Pi_{[0,64]}(q_{A,k}+\Delta q_{A,k}^{\rm del})\right)$. Rate limiting adds actuator memory; travel limits, deadband, and quantization create regions where a change in $v_k^\star$ is not transmitted to $v_k^a$. The claim of Proposition~1 remains exact, but the applied-as-computed premise does not.

Equations~\eqref{eq:coreidentity} and~\eqref{eq:residualsplit} form a comparator system driven by one aggregate residual. Define
\begin{equation}
 \xi_k:=B_pb_k^\star(\delta_k^{\rm par}+\delta_k^{\rm obs}
 +\delta_k^{\rm mag}+\delta_k^{\rm filt}+\delta_k^{\rm act})+r_k^\perp,
 \label{eq:xi}
\end{equation}
so that $z_{k+1}=A_{m,p}z_k+\xi_k$. Iterating this identity gives a conditional input-to-state-style comparison: if $\norm{A_{m,p}^j}\le C_m\gamma^j$ for constants $C_m>0$, $0<\gamma<1$, and if the closed-loop residual satisfies $\norm{\xi_k}\le\bar\xi$, then
\begin{equation}
 \norm{z_k}\le C_m\gamma^k\norm{z_0}+\frac{C_m}{1-\gamma}\bar\xi.
 \label{eq:bound}
\end{equation}
This result is conditional because $\xi_k$ is not exogenous; it depends on adaptation, observation, command processing, actuation, and the nonlinear plant. The bound~\eqref{eq:bound} is therefore an accounting device rather than a nonlinear adaptive stability theorem: it connects the residual decomposition to the saturation, plant variation, observer, deadband, and quantization tests.

\begin{figure*}[!t]
\centering
\includegraphics[width=0.90\textwidth]{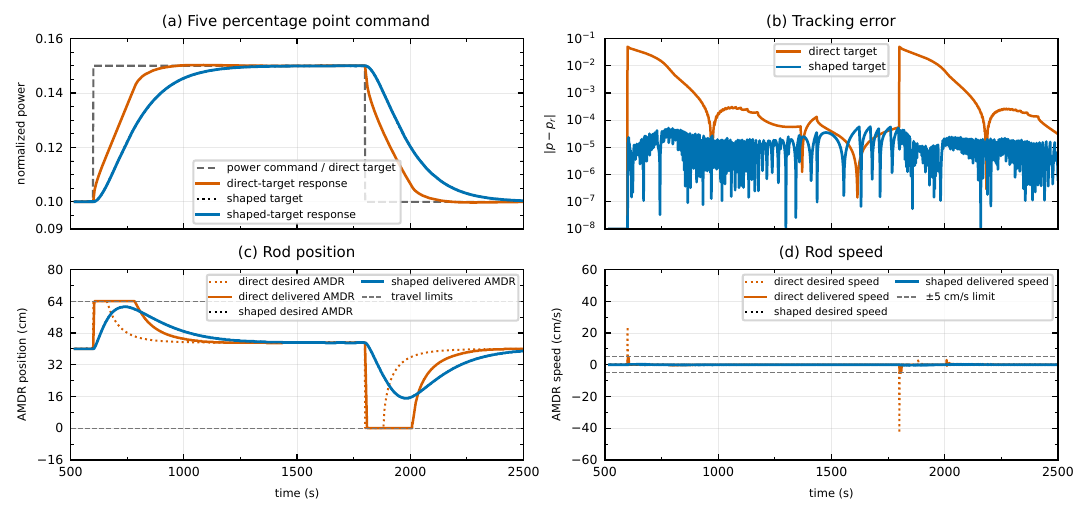}
\caption{Direct and shaped targets for the same $10\%\rightarrow15\%\rightarrow10\%$ command. The direct physical reference is infeasible under the modeled rod speed and travel; the shaped reference remains feasible and reduces tracking error and total rod mileage.}
\label{fig:b7}
\end{figure*}

\section{Evaluation Protocol}
All simulations use the nonlinear ten-state plant, LSODA integration, and a 1 Hz control/actuation clock. Metrics include power root mean square error (RMSE), peak error, settling time, limit exposure, and rod activity. Over $N$ samples, total AMDR mileage is calculated as $M_A:=\sum_{k=0}^{N-1}|q_{A,k+1}-q_{A,k}|$. The reversal count $N_{\rm rev}$ is the number of sign changes in nonzero AMDR motion. Both are actuator duty proxies rather than life predictions, but they relate directly to wear, maintenance exposure, and downstream risk. Our evaluations compare direct and shaped $10\%\rightarrow15\%\rightarrow10\%$ commands; 30 cases in overlapping 9- and 21-case parameter sets; paired adaptive and fixed-nominal runs; power-only feedback and observer initialization; and deadband widths $d_A=0.2$ cm and $d_A=0$. For unexpected RR motion, the 10\% equilibrium is $q_R^{0}=9.4287$ cm; at 600 s the RR moves at no more than $0.9167$ cm/s to $q_R^{0}+\Delta q_R$, $\Delta q_R\in\{-2,-1,1,2\}$ cm, and remains there.

For plant case $j$, let $x_0^{(j)}$ be its equilibrium and $x_{r,0}$ the reference equilibrium. The comparison uses
\begin{equation}
 z_k^{(j)}=S_x\!\left[(x_k^{(j)}-x_0^{(j)})-(x_{r,k}-x_{r,0})\right].
 \label{eq:incremental}
\end{equation}
For the plant family comparisons, online updates are paused during the equilibrium prelude, enabled at the maneuver in adaptive runs, and controller memory is reset before each new plant; execution order therefore transfers no learning. The same $10\%\rightarrow9\%\rightarrow10\%$ case is repeated for each plant $j$, with the same target system, nominal LQR gain, integral action, feedback, and actuator logic. The 9-case reactivity channel grid scales $\alpha_T$ and $\alpha_X$ together by $\{0.8,1.0,1.2\}$ and AMDR differential worth independently by $\{0.9,1.0,1.1\}$. The 21-case expanded set varies $\Lambda$ and thermal gain by $\pm2\%$, $\beta$ and $\lambda_i$ by $\pm1\%$, $\tau_T$ by $\pm5\%$, $\alpha_T$ and $\alpha_X$ by $\pm10\%$, $k_I,k_X,\lambda_I,$ and $\lambda_X$ by $\pm2\%$, $\sigma_X$ by $\pm1\%$, and AMDR/RR differential worth by $\pm5\%$; 16 interior combinations are fixed Latin hypercube samples. Coefficients are fixed within a run. These are simulation ranges, not measured PUR-1 uncertainty bounds or licensing limits.

\section{Results}
\subsection{Reference feasibility and the larger maneuver}
Figure~\ref{fig:b7} compares direct and shaped targets for the same diagnostic $10\%\rightarrow15\%\rightarrow10\%$ command. The direct target requests the full change in one sample: desired AMDR speed reaches 43.1 cm/s, the delivered rod is rate limited for 12 s and travel limited for 391 s, and RMSE is $1.098\times10^{-2}$. The simulated loop nevertheless remains bounded, reaches the high-power plateau, and returns without persistent loss of regulation--an empirical result, not a safety or theorem claim. The shaped target, with natural frequency $\omega_n=0.01$ rad/s, requests 0.29 cm/s, reaches neither limit, lowers RMSE to $1.835\times10^{-5}$, and reduces mileage from 129.8 to 93.1 cm; its slower 585 s settling avoids sustained saturation. In~\eqref{eq:coreidentity}, the direct case primarily enlarges $v_k^a-v_k^{\rm id}$ through $\delta_k^{\rm act}$, whereas shaping keeps the requested reference within available rod authority.


\begin{figure*}[!t]
\centering
\includegraphics[width=0.95\textwidth]{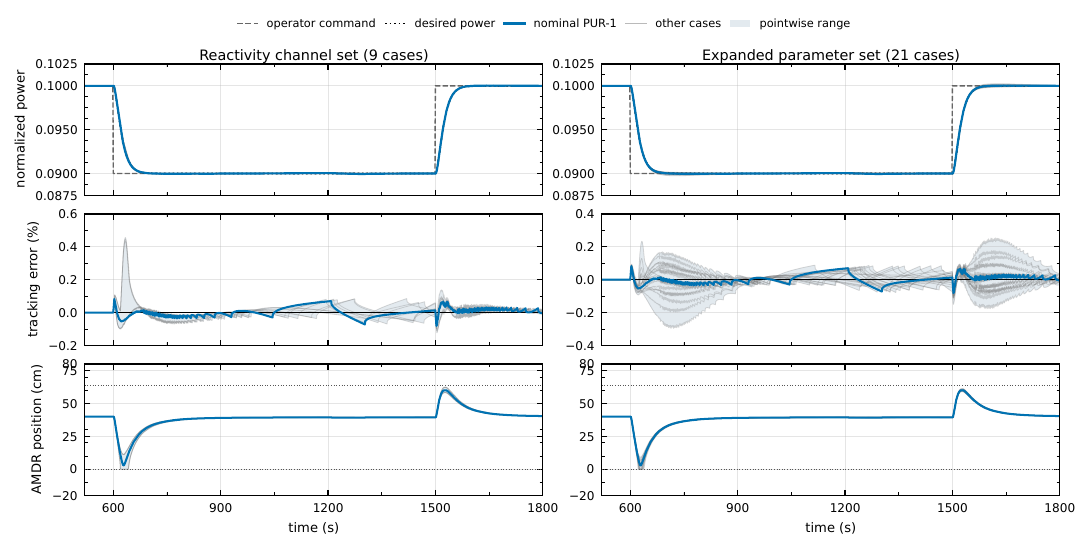}
\caption{Complete adaptive response envelopes for the selected shaped $10\%\rightarrow9\%\rightarrow10\%$ maneuver using~\eqref{eq:incremental} and the parameter ranges in Sec.~IV. The nominal PUR-1 model is emphasized, all other cases remain visible, and shading gives the pointwise range. Similar power histories require different AMDR histories, which our MRAC-LQR implementation effectively utilizes to yield a practically indistinguishable power response across all considered nuclear systems.}
\label{fig:families}
\end{figure*}

\subsection{One nominal design across 30 plant cases}
Figure~\ref{fig:families} shows the adaptive deployment runs for both parameter sets; it is not a frozen-controller robustness envelope. In the 9-case reactivity channel set, settling is 74--76 s, maximum pointwise power spread is $4.74\times10^{-4}$, simultaneous AMDR spread is 16.2 cm, and mileage is 94.7--125.9 cm. For the 21-case expanded set, the corresponding values are 67--191 s, $4.18\times10^{-4}$, 12.5 cm, and 100.5--123.7 cm. Narrow power envelopes with separated rod histories make the adaptive premise observable: one fixed design-time architecture produces comparable output behavior through online, plant-dependent actuation histories. With no accepted LQR redesign, the paired fixed nominal reruns isolate direct adaptation in~\eqref{eq:update} without changing the plant, reference, feedback calculation, integral action, or actuator software. Adaptation lowers RMSE in all nine reactivity channel cases and in 19 of 21 expanded cases; across the overlapping sets, 28 of 30 pairs improve over 29 unique models. The two exceptions remain regulated but favor fixed nominal parameters. This neither proves $r_k^\perp=0$ nor identifies physical coefficients; it shows that matched online adjustment empirically dominates the remaining variation in most sampled cases and supports the adaptive control interpretation of robustness across the stated parameter ranges.


\begin{figure*}[!t]
\centering
\includegraphics[width=0.92\textwidth]{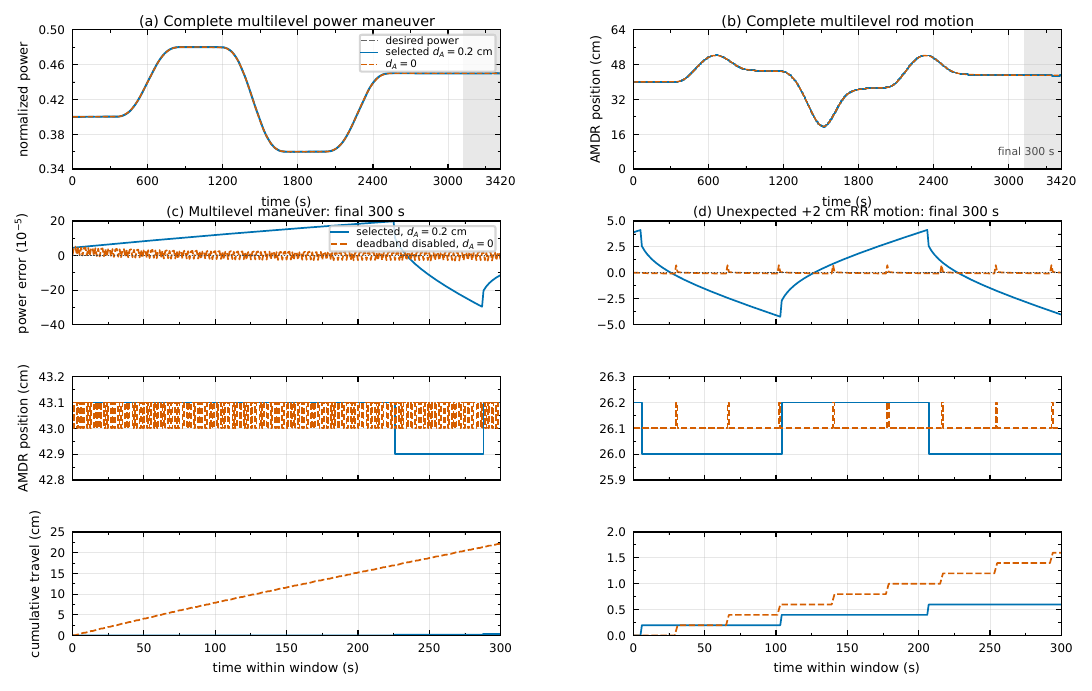}
\caption{Deadband ablation with all other settings fixed. The top row shows the complete multilevel $40\%\rightarrow48\%\rightarrow36\%\rightarrow45\%$ maneuver and rod motion; the shaded interval is expanded below at left. The lower-right column gives the final 300 s after an unexpected, held $+2$ cm RR displacement. Removing the deadband reduces quantization-scale power error but produces frequent 0.1 cm corrections and additional cumulative rod travel.}
\label{fig:deadband_mileage}
\end{figure*}

\subsection{Deadband, precision, and actuator duty}
Figure~\ref{fig:deadband_mileage} first shows the complete $40\%\rightarrow48\%\rightarrow36\%\rightarrow45\%$ maneuver, which exercises both regulation directions over a wider operating range. The selected and no-deadband power responses are nearly identical and neither reaches a rate or travel limit. In the final 300 s, however, removing the deadband raises mileage from 0.4 to 22.2 cm. Over the full maneuver, mileage rises from 87.1 to 115.0 cm and reversals from 4 to 287, while RMSE falls from $9.62\times10^{-5}$ to $2.26\times10^{-5}$. After the persistent $+2$ cm RR displacement, both settings recover in 5 s, but mileage rises from 16.8 to 31.1 cm and reversals from 10 to 160 without the deadband. In~\eqref{eq:coreidentity}, deadband intentionally holds $v_k^a$ fixed for subthreshold corrections, replacing exact small-signal matching by a neighborhood set by deadband and quantization. Delivered input regression avoids treating an undelivered request as plant input, and parameter projection keeps the adaptive coefficients bounded. The multilevel test then supplies the empirical check: no accumulating error, delayed large correction, or escalating rod excursion appears, while quantization-driven switching is strongly suppressed. The mileage and reversal reductions matter because otherwise negligible corrections add motor and drivetrain duty, maintenance exposure, and sensitivity to backlash, stiction, wear, and failure.

\begin{figure*}[!t]
\centering
\includegraphics[width=0.75\textwidth]{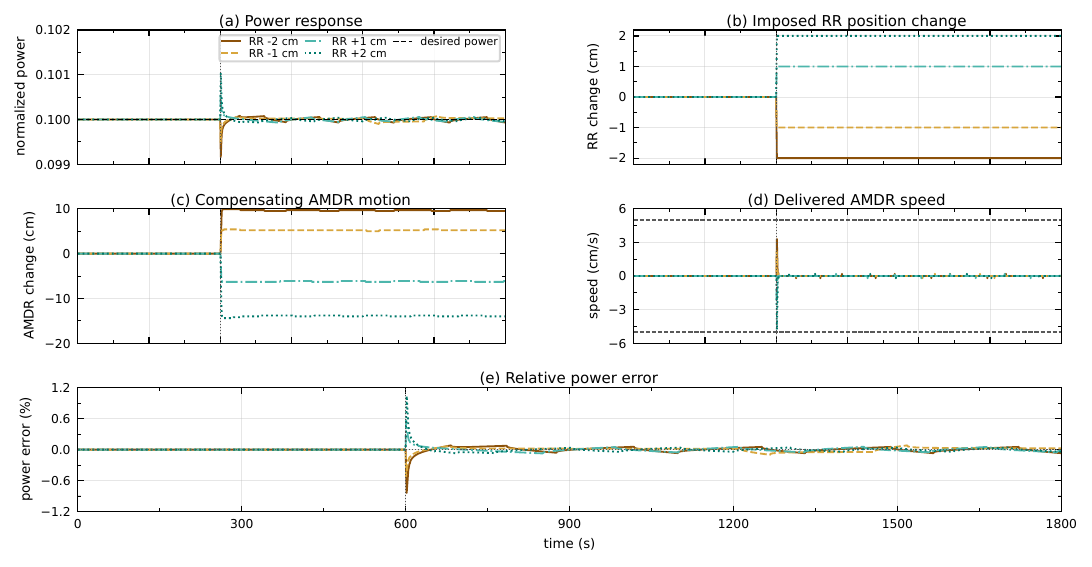}
\caption{Full-horizon responses to sustained RR displacements from the 10\% equilibrium $q_R^{0}=9.4287$ cm to $q_R^{0}+\Delta q_R$, $\Delta q_R=\pm1,\pm2$ cm. Panels show the imposed motion, compensating AMDR action, and power response. At the physical reactivity input level, the imposed RR worth enters the same scalar channel as the AMDR and is therefore a matched disturbance. The independent rod motion and the surrounding nonlinear, filtered, quantized, rate- and travel-limited realization are not covered by the fixed-LTI MRAC--LQR result. These are generic anomaly tests rather than models of a specific failure.}
\label{fig:rr_response}
\end{figure*}

\subsection{Unexpected regulating rod motion and observer boundaries}
At 600 s the RR moves from $q_R^{0}=9.4287$ cm to $q_R^{0}+\Delta q_R$ and is held while MRAC--LQR commands only the AMDR. Figure~\ref{fig:rr_response} shows all four signed cases. The motions add approximately $-6.83$, $-3.61$, $+3.98$, and $+8.33$ pcm; all recover in 2--5 s without an AMDR limit, with compensating motion and power-error sign reversing with the anomaly.

With exact, power-only feedback, warm initialization retains the 5 s recovery of the $+2$ cm case and tracks the shaped $10\%\rightarrow9\%\rightarrow10\%$ maneuver nearly as in full state feedback. The compact comparison in Table~\ref{tab:observer_init} quantifies all three initialization cases. The tenfold-xenon case remains close to the warm response even though $\widehat X/X=9.887$ at 600 s; in the observer-cold case, $\widehat X/X=1.63\times10^{-4}$ at the same time and the inconsistent precursor and temperature memories produce a large engagement transient. A $\pm1\%$ sensor bias produces about a 1\% physical offset. 

\par\smallskip
\noindent\begin{minipage}{\columnwidth}
\centering
\refstepcounter{table}\label{tab:observer_init}
{\footnotesize\textsc{Table \thetable}}\\[-1pt]
{\scriptsize\textsc{Closed-loop observer initialization results. All cases begin at the true 10\% equilibrium with exact measured power.}}
\scriptsize
\setlength{\tabcolsep}{2.4pt}
\begin{tabular}{@{}lcccc@{}}
\toprule
Initialization & Peak prelude & AMDR range & Limit time & RMSE, \\
 & power excursion & in prelude (cm) & rate/travel (s) & $600$--$2500$ s \\
\midrule
Warm/exact & $-$ & $0.0$ & $0/0$ & $2.69\times10^{-5}$ \\
$\widehat X_0=10X_0$ & $0.0622\%$ & $0.6$ & $0/0$ & $2.64\times10^{-5}$ \\
Observer-cold & $11.52\%$ & $47.0$ & $10/49$ & $4.58\times10^{-5}$ \\
\bottomrule
\end{tabular}
\end{minipage}
\par\smallskip

\section{Conclusion}
In this work, we showed how MRAC--LQR can operate with nonlinear reactor dynamics and constrained actuation, including saturation, speed limits, and quantization. The plant family results demonstrate direct adaptation within prescribed bounds while retaining the nominal LQR design, reference construction, and actuator logic. This separation makes the scope of online change explicit for subsequent verification and regulatory review, although bounded coefficients alone do not establish stability of the complete constrained loop. The pool research reactor setting provides a path from simulation to experimental demonstrations and broader consideration in advanced reactor applications.

\end{document}